\documentclass[sigconf]{acmart}
\usepackage{booktabs} 
\usepackage{graphicx}
\usepackage{wrapfig}
\usepackage{tabularx}
\usepackage{arydshln}
\usepackage{amsmath}
\usepackage{breakurl}
\usepackage{multirow}
\usepackage{amsmath}
\usepackage{amsfonts}
\usepackage{amssymb}

\usepackage{algorithm, algpseudocode}
\usepackage{paralist} 
\usepackage{color}
\usepackage{rotating}
\usepackage[mathscr]{eucal} 
\usepackage{amsbsy} 
\usepackage{booktabs}
\usepackage{xcolor}
\usepackage{tikz}
\usetikzlibrary{snakes}
\usepackage{epstopdf}
\usepackage{xspace}
\usepackage{balance}
\usepackage{tablefootnote}
\usepackage{empheq} 
\usepackage{moresize}
\usepackage{lipsum}
\usepackage[font={bf,small}]{caption}
\usepackage[inline]{enumitem}
\setlist{nolistsep}
\usepackage{hyperref}
\usepackage{subfig}
\usepackage{pifont}
\newcommand{\cmark}{\ding{51}}%
\newcommand{\xmark}{\ding{55}}%

\usepackage{dsfont}
\usepackage{comment}

\newcounter{ALC@tempcntr}

\newcommand{\W}{\mathbf{W}}
\newcommand{\matC}{\mathbf{C}}

\newcommand{\matU}{\mathbf{U}}
\newcommand{\Y}{\mathbf{Y}}
\newcommand{\Z}{\mathbf{Z}}
\newcommand{\matV}{\mathbf{V}}
\newcommand{\matA}{\mathbf{A}}
\newcommand{\matW}{\mathbf{W}}
\newcommand{\matM}{\mathbf{M}}
\newcommand{\matY}{\mathbf{Y}}

\newcommand{\perm}{\mathbf{P}}

\newcommand{\graphG}{G}

\newcommand{\V}{\mathcal{V}}
\newcommand{\E}{\mathcal{E}}
\newcommand{\R}{\mathcal{R}_u^k}
\newcommand{\matD}{\mathbf{D}}
\newcommand{\A}{\mathbf{A}}

\newcommand{\matS}{\mathbf{S}}
\newcommand{\matSigma}{\mathbf{\Sigma}}

\newcommand{\y}{\mathbf{y}}
\newcommand{\vecdk}{\mathbf{d}^k_u}
\newcommand{\vecd}{\mathbf{d}}
\newcommand{\veca}{\mathbf{f}}
\newcommand{\vecf}{\mathbf{f}}
\newcommand{\context}{\mathbf{c}}
\newcommand{\method}{REGAL\xspace}
\newcommand{\embedding}{xNetMF\xspace}
\newcommand{\nystrom}{Nystr{\"o}m }
\newcommand{\Ell}{\mathcal{L}}

\DeclareCaptionLabelFormat{andtable}{#1~#2  \&  \tablename~\thetable}

\newtheorem{problem}{Problem}

\author{Mark Heimann}
\affiliation{\institution{University of Michigan, Ann Arbor}}
\email{mheimann@umich.edu}

\author{Haoming Shen}
\affiliation{\institution{University of Michigan, Ann Arbor}}
\email{hmshen@umich.edu}

\author{Tara Safavi}
\affiliation{\institution{University of Michigan, Ann Arbor}}
\email{tsafavi@umich.edu}

\author{Danai Koutra}
\affiliation{\institution{University of Michigan, Ann Arbor}}
\email{dkoutra@umich.edu}

\setcopyright{none}
\begin{document}

\acmConference[CIKM '18]{The 27th ACM International Conference on Information and Knowledge Management}{October 22--26, 2018}{Torino, Italy}
\acmBooktitle{The 27th ACM International Conference on Information and Knowledge Management (CIKM '18), October 22--26, 2018, Torino, Italy}
\acmPrice{15.00}
\acmDOI{10.1145/3269206.3271788}
\acmISBN{978-1-4503-6014-2/18/10}
  
\begin{abstract}
Problems involving multiple networks are prevalent in many scientific and other domains.
In particular, network alignment, or the task of identifying corresponding nodes in different networks, has applications across the social and natural sciences.
Motivated by recent advancements in node representation learning for \emph{single}-graph tasks, 
we propose \method (REpresentation learning-based Graph ALignment), a framework that leverages the power of automatically-learned node representations to match nodes across \emph{different} graphs.
Within \method we devise \embedding, an elegant and principled node embedding formulation that uniquely generalizes to multi-network problems.
Our results demonstrate the utility and promise of unsupervised representation learning-based network alignment in terms of both speed and accuracy.
\method runs up to $30\times$ faster in the representation learning stage than comparable methods, outperforms existing network alignment methods by 20 to 30\% accuracy on average, and scales to networks with millions of nodes each.

\end{abstract}



\title{REGAL: Representation Learning-based Graph Alignment}
\maketitle

\section{Introduction}
Networks are powerful structures that naturally capture the wealth of relationships in our interconnected world, such as co-authorships, email exchanges, and friendships \cite{koutra2017individual}.
The data mining community has accordingly proposed various methods for numerous tasks over a \emph{single} network, like anomaly detection, link prediction, and user modeling.
However, many graph mining tasks involve joint analysis of nodes across \emph{multiple} networks.
Some problems, like network alignment~\cite{netalign,final,bigalign} and graph similarity~\cite{deltacon}, are inherently defined in terms of multiple graphs.  
In other cases, it is desirable to perform analysis across a collection of graphs, such as the MRI-based brain graphs of patients~\cite{brain-networks}, or snapshots of a temporal graph~\cite{timecrunch}. 

\begin{figure}[t!]
    \centering
    \includegraphics[width=0.49\textwidth]{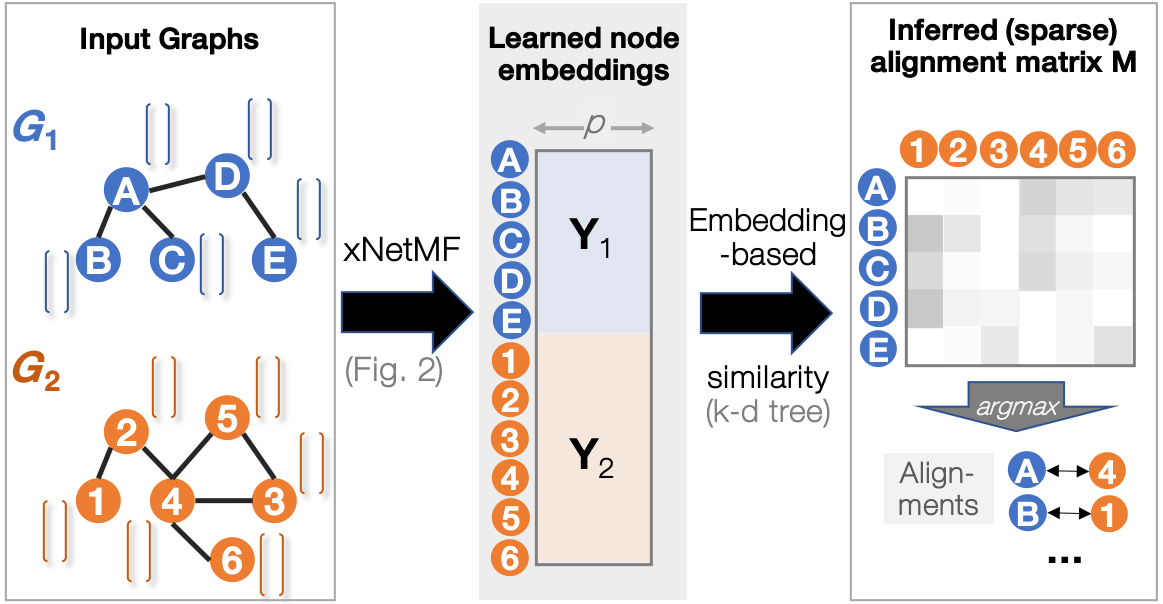}
        \vspace{-0.5cm}
    \caption{Pipeline of proposed graph alignment method, \method, based on our \embedding representation learning method.}
    \label{pipeline}
    \vspace{-0.4cm}
\end{figure}

In this work, we study  \textbf{network alignment} or matching, which is the problem of finding corresponding nodes in different networks.
Network alignment is crucial for identifying similar users in different social networks, analyzing chemical compounds, studying protein-protein interaction, and various computer vision tasks, among others \cite{netalign}.  
Many existing methods try to relax the computationally hard optimization problem, as \textit{designing features} that can be directly compared for nodes in different networks is not an easy task. 
However, recent advances~\cite{node2vec,deepwalk,line,SDNE} 
have automated the process of learning node feature representations and have led to state-of-the-art performance in downstream prediction, classification,  and clustering tasks.
Motivated by these successes, we propose network alignment via matching latent, learned node representations.
Formally, the problem can be stated as: 

\begin{problem}
Given two graphs $G_1$ and $G_2$ with node-sets $\mathcal{V}_1$ and $\mathcal{V}_2$ and possibly node attributes $\mathcal{A}_1$ and $\mathcal{A}_2$ resp., devise an efficient \textbf{network alignment method} that aligns nodes by learning \textbf{directly comparable} node representations $\matY_1$ and $\matY_2$, from which a node mapping $\phi: \mathcal{V}_1 \rightarrow \mathcal{V}_2$ between the networks can be inferred.
\label{prob}
\end{problem}

To this end, we introduce \textbf{\method}, or REpresentation-based Graph ALignment, a framework that efficiently identifies node matchings by greedily aligning their latent feature representations (Fig.~\ref{pipeline}).
\method is both highly intuitive and extremely powerful given suitable node feature representations.  
For use within this framework, we propose Cross-Network Matrix Factorization (\textbf{\embedding}), which we introduce specifically to satisfy the requirements of the task at hand.
\embedding differs from most existing representation learning approaches that (i)~rely on proximity of nodes in a \textit{single} graph, yielding embeddings that are not comparable across disjoint networks~\cite{mlg_paper}, and (ii)~often involve some procedural randomness (e.g., random walks), which introduces variance in the embedding learning, even in one network.
By contrast, \embedding preserves \textit{structural} similarities rather than proximity-based similarities, allowing for generalization beyond a single network.

To learn node representations through an efficient, low-variance process, we formulate \embedding as matrix factorization over a similarity matrix that incorporates structural similarity 
and attribute agreement (if the latter is available) between nodes in disjoint graphs. 
To avoid explicitly constructing a full similarity matrix, which requires computing all pairs of similarities between nodes in the multiple input networks, we extend the \nystrom low-rank approximation commonly used for large-scale kernel machines ~\cite{nystrom}.
\embedding is thus a principled and efficient \emph{implicit} matrix factorization-based approach, requiring a fraction of the time and space of the na\"{i}ve approach while avoiding ad-hoc sparsification heuristics.

Our contributions may be stated as follows: \setlist[itemize]{leftmargin=*}
\begin{itemize}
    \item \textbf{Problem Formulation}. We formulate the important unsupervised graph alignment problem as a problem of learning and matching node representations that \emph{generalize to multiple graphs}. 
    To the best of our knowledge, we are the first to do so.
    \item  \textbf{Principled Algorithms}. We introduce a flexible alignment framework, \method (Fig.~\ref{pipeline}), which  learns node alignments by jointly embedding multiple graphs and comparing the most similar embeddings across graphs \emph{without} performing all pairwise comparisons. 
    Within \method we devise \embedding, an elegant and principled representation learning formulation.
    \embedding learns embeddings from structural and, if available, attribute identity, which are characteristics most conducive to multi-network analysis.
    \item \textbf{Extensive Experiments}. 
    Our results demonstrate the utility of representation learning-based network alignment in terms of both speed and accuracy. 
    Experiments on real graphs show that \embedding runs up to $30\times$ faster than several existing network embedding techniques, and \method outperforms traditional network alignment methods by 20-30\% in accuracy.
\end{itemize}

For reproducibility, the source code of \method and \embedding is publicly available at \url{https://github.com/GemsLab/REGAL}.

\section{Related Work}
\label{sec:related}
Our work focuses on the problem of network alignment, and is related to node representation learning and matrix approximation. 

\vspace{.15cm}
\noindent\textbf{Network Alignment.} 
Instances of the network alignment or matching problem appear in various settings: from data mining to security and re-identification~\cite{final,bigalign,netalign}, chemistry, bioinformatics~\cite{multimagna,isorank,klau}, databases, translation~\cite{netalign}, vision, and pattern recognition~\cite{ZaslavskiyBV09}. Network alignment is usually formulated as the optimization problem
$\min_\mathbf{P}||\mathbf{P}\A_1\mathbf{P}^T - \A_2||_F^2$~\cite{bigalign}, where $\mathbf{A}_1$ and $\mathbf{A}_2$ are the adjacency matrices of the two networks to be aligned, and $\mathbf{P}$ is a permutation matrix or a relaxed version thereof, such as doubly stochastic matrix~\cite{VogelsteinCPKFVP11} or some other concave/convex relaxation~\cite{ZaslavskiyBV09}. Popular proposed solutions to the network alignment problem span genetic algorithms, spectral methods, clustering algorithms, decision trees, expectation maximization,  probabilistic approaches, and distributed belief propagation \cite{multimagna,isorank,klau,netalign}. These methods usually require carefully tailoring for special formats or properties of the input graphs.  For instance, specialized formulations may be used when the graphs are bipartite \cite{bigalign} or contain node/edge attributes \cite{final}, or when some ``seed'' alignments are known a priori \cite{seedalign}. Prior work using node embeddings designed for social networks to align users \cite{liu2016aligning} has required such seed alignments.  In contrast, our approach can be applied to attributed and unattributed graphs with virtually no change in formulation, and is \emph{unsupervised}: it does not require prior alignment information to find high-quality matchings. Recent work \cite{heimann2018hashalign} has used hand-engineered features, while our proposed approach leverages the power of latent feature representations.  

\vspace{.15cm}
\noindent\textbf{Node Representation Learning.} Representation learning methods try to find similar embeddings for similar nodes~\cite{gemsurvey}.  
They may be based on shallow~\cite{node2vec} or deep architectures~\cite{SDNE}, and may discern neighborhood structure through random walks \cite{deepwalk} or first- and second-order connections \cite{line}. 
Recent work inductively learns representations \cite{graphsage} and/or incorporates textual or other node attributes \cite{lane, tadw}.
However, all these methods use node \emph{proximity} or neighborhood overlap to drive embedding, which has been shown to lead to inconsistency \emph{across} networks \cite{mlg_paper}.

Unlike these methods, the recent work struc2vec \cite{struc2vec} 
preserves \emph{structural} similarity of nodes, regardless of their proximity in the network. Prior to this work, existing methods for structural role discovery mainly focused on hand-engineered features \cite{rolediscovery}.  
However, for structurally similar nodes, struc2vec embeddings were found to be visually more comparable \cite{struc2vec} than those learned by state-of-the-art proximity-based node embedding techniques as well as existing methods for role discovery \cite{rolx}.  
While this work is most closely related to our proposed node embedding method, we summarize some crucial differences in Table \ref{tab:structure_comp}.  Additionally, we note that struc2vec, like work on structural node embeddings concurrent to ours~\cite{graphwave}, cannot natively use node attributes.

\begin{table}[t!]
     \caption{Qualitative comparison of structure-based embeddings.}
     \label{tab:structure_comp}
     \vspace{-0.3cm}
{\footnotesize
\begin{tabular}{ p{3.75cm} p{3.75cm}}
     \toprule
     \textbf{struc2vec~\cite{struc2vec}} &  \textbf{\embedding (Proposed)}\\
     \midrule
     Variable-length degree sequences compared with dynamic time warping & Fixed length vectors capturing neighborhood degree distributions \\ \midrule
     Variance-inducing, time-consuming random walk-based sampling & Efficient matrix factorization \\ \midrule
     Heuristic-based omission of similarity computations & Low-rank implicit approximation of full similarity matrix \\ \hline

     >0.5 hours to embed Arxiv network (Table \ref{tab:datasets}) \cite{ppi} using optimizations & <90 sec to embed Arxiv network; $\sim22\times$ speedup \\
     \bottomrule
    \end{tabular}
    }
    \vspace{-0.2cm}
    \end{table}
    
\begin{table}[t!]
 \vspace{-0.1cm}
\caption{Qualitative comparison of related work to the embedding module of \method. (*: Method not based on random walks, RW)
} 
 \vspace{-0.3cm}
\label{tab:qualitative}
\centering
\resizebox{\columnwidth}{!}{
\begin{tabular}{lcccc|c}
\toprule
  \textbf{ } & \rotatebox[origin=c]{0}{\textbf{\bf Structure}} & \rotatebox[origin=c]{0}{\textbf{\bf Attributes}} & \rotatebox[origin=c]{0}{\textbf{\bf RW-free* }} &  \rotatebox[origin=c]{0}{\textbf{\bf Scalable}} & 
  \rotatebox[origin=c]{0}{\textbf{\bf Cross-net}} \\ \hline 
  LINE~\cite{line} &  \xmark & \xmark & \cmark & \cmark & \xmark \\
  TADW~\cite{tadw} &  \xmark & \cmark & \cmark & \textbf{?} & \xmark \\
  node2vec~\cite{node2vec} & \textbf{?} & \xmark & \xmark & \textbf{?} & \xmark \\ 
  struc2vec ~\cite{struc2vec} & \cmark & \xmark & \xmark & \xmark & ? \\ \hline
  \embedding (in \method)  & \cmark & \cmark & \cmark & \cmark & \cmark \\
\bottomrule
\end{tabular}
}
 \vspace{-0.4cm}
\end{table}

Many well-known node embedding methods based on shallow architectures such as the popular skip-gram with negative sampling (SGNS) have been cast in matrix factorization frameworks \cite{tadw, netmf}.  
However, ours is the first to cast node embedding using SGNS to capture \emph{structural} identity in such a framework.  
In Table \ref{tab:qualitative} we extend our qualitative comparison to some other well-known methods that use similar architectures.  
Their limitations inspire many of our choices in the design of \method and \embedding.

In terms of applications, very few works consider using learned representations for problems that are inherently defined in terms of multiple networks, where embeddings must be compared. \cite{embed-graphsim} computes a similarity measure between graphs based on the Earth Mover's Distance \cite{emd} between simple node embeddings generated from the eigendecomposition of the adjacency matrix. Here, we consider the significantly harder problem of learning embeddings that may be individually matched to infer node-level alignments.
 
\vspace{.15cm}
\noindent\textbf{Low-Rank Matrix Approximation.}  The \nystrom method has been used for low-rank approximations of large, dense similarity matrices \cite{nystrom}. 
While the quality of its approximation has been extensively studied theoretically and empirically in a statistical learning context for kernel machines \cite{nystrom-guarantees}, to the best of our knowledge it has not been considered in the context of node embedding.

\begin{figure*}[t!]
	\centering
	\vspace{-0.2cm}
	\includegraphics[width=0.92\textwidth]{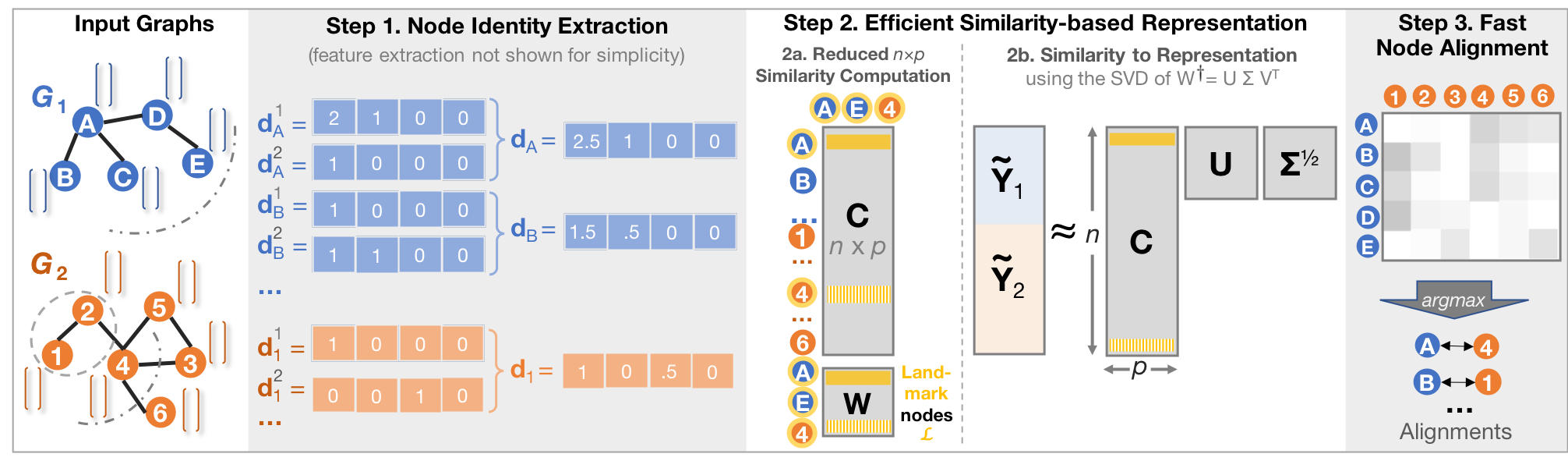}
	\vspace{-0.4cm}
	\caption{Proposed \method approach, consisting of 3 main steps. In the example, for the structural identity, up to $K=2$ hop away neighborhoods are taken into account (the 1-hop and 2-hop neighborhoods for nodes $A$ and $1$ are shown with dashed and dash-dotted lines, respectively). The discount factor is set to $\delta=0.5$.
	For simplicity, no logarithmic binning is applied on $\vecdk$. 
} 

	\label{embed-pipeline}
	\vspace{-0.3cm}
\end{figure*}

\section{\method: REpresentation Learning-based Graph ALignment}

In this section we introduce our representation learning-based network alignment framework, \method, for Problem~\ref{prob}.
For simplicity we focus on aligning two graphs (e.g., social or protein networks), though our method can easily be extended to more networks. 
Let $G_1(\V_1, \E_1)$ and $G_2(\V_2, \E_2)$ be two unweighted and undirected graphs with node sets $\mathcal{V}_1$ and $\mathcal{V}_2$; edge sets $\mathcal{E}_1$ and $\mathcal{E}_2$; and  
possibly node attributes $\mathcal{A}_1$ and $\mathcal{A}_2$, respectively.  Note that these graphs do \emph{not} have to be the same size, unlike many other network alignment formulations that have this (often unrealistic) restriction. 
Let $n$ be the number of nodes across graphs, i.e., $n=|\mathcal{V}_1| + |\mathcal{V}_2|$.
We define the main symbols   in Table~\ref{tab:dfn}.

The steps of \method may be summarized as:
\setlist[enumerate]{leftmargin=*}
\begin{enumerate}
	\item \textbf{Node Identity Extraction}: The first step extracts structure- and attribute-related information for all $n$ nodes.
	\item \textbf{Efficient Similarity-based Representation}: The second step obtains the node embeddings, conceptually by factorizing a similarity matrix of the node identities from the previous step.
	To avoid the expensive computation of pairwise node similarities and explicit factorization, we extend the \nystrom method for low-rank matrix approximation to perform an \emph{implicit} similarity matrix factorization by \textbf{(a)} comparing the similarity of each node only to a sample of $p \ll n$ ``landmark'' nodes, and \textbf{(b)} using these node-to-landmark similarities to construct our representations from a decomposition of its low-rank approximation.
	\item \textbf{Fast Node Representation Alignment}: Finally, we align nodes between graphs by greedily matching the embeddings with an efficient data structure that allows for fast identification of the top-$\alpha$ most similar embeddings from the other graph(s). 
\end{enumerate}

In the rest of this section we discuss and justify each step of \method, the pseudocode of which is given in Algorithm~\ref{regal}.  Note that the first two steps, which output a set of node embeddings, comprise our \embedding method, which may be independently used, particularly for further cross-network analysis tasks.

\begin{table}[t!]
     \caption{Major symbols and definitions.}
     \label{tab:dfn}
     \vspace{-0.2cm}
{\footnotesize
\begin{tabular}{ p{1.64cm} p{6.1cm}}
     \toprule
     \textbf{Symbols} &  \textbf{Definitions}\\
     \midrule
     $\graphG_i(\V_i, \E_i, \mathcal{A}_i)$ & graph $i$ with nodeset $\V_i$, edgeset $\E_i$, and node attributes $\mathcal{A}_i$ \\ 
     $\matA_i$ & adjacency matrix of $\graphG_i$ \\ 
     $n_i$ & number of nodes in graph $\graphG_i$ \\
     $\V = \V_1 \cup \V_2$
    & combined set of vertices in  $\graphG_1$ and $\graphG_2$ \\
    $|\V| = n$  &  total number of nodes in graphs $\graphG_1$ and $\graphG_2$ \\
     $d_{\text{avg}}$ & average node degree \\ \hline

     ${\R}$ & set of $k$-hop neighbors of node $u$\\  
      $\vecdk$ & vector of node degrees in a single set $\R$   \\
      $K$ & maximum hop distance considered \\
      $\delta$ & discount factor in $(0,1]$ for distant neighbors  \\
      $\vecd_u$ & $=\sum_{k=1}^{K} \delta^{k-1} \vecdk$ combined neighbor degree vector for node $u$ \\ 
      $b$ & number of buckets for degree binning \\
      $\vecf_u$ & $F$-dimensional attribute vector for node $u$ \\
      \hline
     
      $\matS, \tilde{\matS}$ & combined structural and attribute-based similarity matrix, and its approximation \\
      $\matY, \tilde{\matY}$ & matrix with node embeddings as rows, and its approximation \\
     $p$ & number of landmark nodes in \method\\ 
      
     $\alpha$ & the number of alignments to find per node \\
     
     \bottomrule
    \end{tabular}
    }
    \vspace{-0.3cm}
    \end{table}
    \setlength{\textfloatsep}{1\baselineskip plus 0.2\baselineskip minus 0.5\baselineskip}

\subsection{Step 1: Node Identity Extraction}
\label{sec:step1}
The goal of \method's representation learning module, \embedding, is to define node ``identity'' in a way that generalizes to multi-network problems.
This step is critical because many existing works define identity based on node-to-node proximity, but in multi-network problems nodes have no direct connections to each other and thus cannot be sampled in each other's contexts by random walks on separate graphs. 
To overcome this problem, we focus instead on more broadly comparable, generalizable quantities: \emph{structural} identity, which relates to structural roles~\cite{rolx}, and \emph{attribute}-based identity.

\vspace{0.1cm}
\noindent \textbf{Structural Identity}. In network alignment, the well-established assumption is that aligned nodes have similar \textit{structural} connectivity or degrees \cite{bigalign,final}.
Adhering to this assumption, we propose to learn about a node's structural identity from the degrees of its neighbors.  
To gain higher-order information, we also consider neighbors up to $k$ hops from the original node.  

For a node $u \in \V$, we denote $\R$ as the set of nodes that are exactly $k \ge 0$ steps away from $u$ in its graph $\graphG_i$.   
We want to capture degree information about the nodes in $\R$.  A basic approach would be to store the degrees in a $D$-dimensional vector $\vecdk$, where $D$ is the maximum degree in the original graph $\graphG$, with the $i$-th entry of $\vecdk$, or $d_u^k(i)$, the number of nodes in $\R$ with degree $i$.  For simplicity, an example of this approach is shown for the vectors $\vecd_A, \vecd_B$, etc.\ in Fig.\ \ref{embed-pipeline}. 
However, real graphs have skewed degree distributions.
To prevent one high-degree node from inflating the length of these vectors, we bin nodes together into $b = \lceil \log_2D \rceil$ logarithmically scaled buckets such that the $i$-th entry of $\vecdk$ contains the number of nodes $u \in \R$ such that $\lfloor \log_2(deg(u)) \rfloor = i$.  This has two benefits: (1) it shortens the vectors $\vecdk$ to a manageable $\lceil \log_2D \rceil$ dimensions, and (2) it makes their entries more robust to small changes in degree introduced by noise, especially for high degrees when more different degree values are combined into one bucket.

\vspace{0.1cm}
\noindent  \textbf{Attribute-Based Identity}.
Node attributes, or features, have been shown to be useful for cross-network tasks~\cite{final}.  
Given $F$ node attributes, we can create for each node $u$ an $F$-dimensional vector $\vecf_u$ representing its values (or lack thereof).
For example,  $f_u(i)$ corresponds to the $i^{th}$ attribute value for node $u$.  
Since we focus on node representations, we mainly consider node attributes, although we note that statistics such as the mean or standard deviation of edge attributes on incident edges to a node can easily be turned into node attributes.
Note that while \method is flexible to incorporate attributes, if available, it can also rely solely on structural information when such side information is not available. 

\vspace{0.1cm}
\noindent \textbf{Cross-Network Node Similarity}. We now incorporate the above aspects of node identity into a combined similarity function that can be used to compare nodes within \emph{or across} graphs, relying  on the comparable notions of structural and attribute identity, rather than direct proximity of any kind: 
\setlength\abovedisplayskip{5pt}
\begin{equation}
\text{sim}(u,v) = \exp{[-\gamma_s \cdot ||\vecd_u - \vecd_v||_2^2 - \gamma_a \cdot \text{dist}(\vecf_u,\vecf_v) ], }  
\label{eq:sim}
\setlength\belowdisplayskip{5pt}
\end{equation}
\noindent where $\gamma_s$ and $\gamma_a$ are scalar parameters controlling the effect of the structural and attribute-based identity respectively; $\text{dist}(\vecf_u,\vecf_v)$ is the attribute-based distance of nodes $u$ and $v$, discussed below (this term is ignored if there are no attributes); $\vecd_u = \sum_{k=1}^{K} \delta^{k-1} \vecdk$ is the neighbor degree vector for node $u$ aggregated over $K$ different hops;  
$\delta \in (0,1]$ is a discount factor for greater hop distances; and $K$ is a maximum hop distance to consider (up to the graph diameter).  Thus, we compare structural identity at several levels by combining the neighborhood degree distributions at several hop distances, attenuating the influence of distant neighborhoods with a weighting schema that is often encountered in diffusion processes~\cite{deltacon}. 

The distance between attribute vectors depends on the type of node attributes (e.g., categorical, real-valued).
A variety of functions can be employed accordingly. 
For categorical attributes, which have been studied in attributed network alignment~\cite{final}, we propose using the number of disagreeing features as a attribute-based distance measure of nodes $u$ and $v$:
$\text{dist}(\veca_u, \veca_v) = \textstyle \sum_{i=1}^{F} \mathds{1}_{f_u(i) \neq f_v(i)},$
where $\mathds{1}$ is the indicator function.  Real-valued attributes can be compared by Euclidean or cosine distance, for example. 

\subsection{Step 2: Efficient Similarity-based Representation}

As we have mentioned, many representation learning methods are stochastic~\cite{line,deepwalk,SDNE,node2vec,struc2vec}.
A subset of these rely on random walks on the original graph~\cite{deepwalk,node2vec} or a generated multi-layer similarity graph~\cite{struc2vec}) to sample context for the SGNS embedding model.  
For cross-network analysis, we avoid random walks for two reasons:  
(1) The variance they introduce in the representation learning often makes embeddings across different networks non-comparable~\cite{mlg_paper};  
and
(2) they can add to the computational expense. For example, node2vec's total runtime is dominated by its sampling time \cite{node2vec}.

To overcome the aforementioned issues,  we propose a new \emph{implicit} matrix factorization-based approach that leverages a combined structural and attribute-based similarity matrix $\matS$, which is induced by our similarity function in Eq.~\eqref{eq:sim} and considers affinities at different neighborhoods.  
Intuitively, the goal is to find $n \times p$ matrices $\matY$ and $\Z$ such that:
$\matS \approx \matY \Z^\top$, where $\matY$ is the node embedding matrix and $\Z$ is not needed for our purposes. We first discuss the limitations of traditional approaches, then propose an efficient way of obtaining the embeddings \textit{without} ever explicitly computing $\matS$.

\vspace{0.1cm}
\noindent \textbf{Limitations of Existing Approaches}.
A natural but na{\"i}ve approach is to compute combined structural and attribute-based similarities between \emph{all} pairs of nodes within and across \emph{both} graphs to form the matrix $\matS$, such that $\matS_{ij} = sim(i,j) \; \forall i,j \in \V$.  
Then $\matS$ can be explicitly factorized, for example by minimizing a factorization loss function given $\matS$ as input, (e.g., the Frobenius norm $||\matS - \matY \Z^\top||_F^2$ \cite{nmf}).
However, both the computation and storage of $\matS$ have \emph{quadratic} complexity in $n$.
While this would allow us to embed graphs jointly, it lacks the needed scalability for multiple large networks.

Another alternative is to create a \emph{sparse} similarity matrix by calculating only the ``most important'' similarities, for each node choosing a small number of comparisons using heuristics like similarity of node degree \cite{struc2vec}.  
However, such ad-hoc heuristics may be fragile in the context of noise.  We will have no approximation at all for most of the similarities, and there is no guarantee that the most important ones are computed. 

\vspace{0.1cm}
\noindent \textbf{Step 2a: Reduced $n \times p$ Similarity Computation.}
Instead, we propose a principled way of \textit{approximating} the full similarity matrix $\matS$ with a low-rank matrix $\tilde{\matS}$, which is \textit{never} explicitly computed.  
To do so, we randomly select $p\ll n$ ``\emph{landmark}'' nodes chosen across both graphs $\graphG_1$ and $\graphG_2$ and compute their similarities to all $n$ nodes in these graphs using Eq.~\eqref{eq:sim}.
This yields an $n \times p$ similarity matrix $\matC$, from which we can extract a $p \times p$ ``landmark-to-landmark'' submatrix $\matW$.
As we explain below, these two matrices suffice to approximate the full similarity matrix and allow us to obtain node embeddings \emph{without} actually computing and factorizing $\tilde{\matS}$. 

To do so, we extend the \nystrom method, which has applications in randomized matrix methods for kernel machines~\cite{nystrom}, to node embedding. 
The low-rank matrix $\tilde{\matS}$ is given as: 

\begin{equation}
\tilde{\matS} = \matC\W^{\dagger}\matC^\top,
\label{eq:approx}
\end{equation}

where $\matC$ is an $n \times p$ matrix formed by sampling $p$ landmark nodes from $\V$  and computing the similarity of all $n$ nodes of $\graphG_1$ and $\graphG_2$ to the $p$ landmarks only, as shown in Fig.~\ref{embed-pipeline}.  Meanwhile, $\W^{\dagger}$ is the pseudoinverse of $\matW$, a $p \times p$ matrix consisting of the pairwise similarities among the landmark nodes (it corresponds to a subset of $p$ rows of $\matC$). We choose landmarks randomly; more elaborate (and slower) sampling techniques based on leverage scores \cite{nystrom-guarantees} or node centrality measures offer little, if any, performance improvement. 

Because $\tilde{\matS}$ contains an estimate for the similarity between any pair of nodes in either graph, it would still take $\Omega(n^2)$ time and space to compute and store.  
However, as we discuss below, to learn node representations we \emph{never} have to explicitly construct $\tilde{\matS}$ either.

\vspace{0.1cm}
\noindent \textbf{Step 2b: From Similarity to Representation.}  
Recall that our ultimate interest is not in the similarity matrix $\matS$ or even an approximation such as $\tilde{\matS}$, but in the node embeddings that we can obtain from a factorization of the latter.
We now show that we can actually obtain these from the decomposition in Eq.~\eqref{eq:approx}:

\begin{theorem}
	Given graphs $\graphG_1(\V_1,\E_1)$ and $\graphG_2(\V_2,\E_2)$ with $n\times n$ joint combined structural and attribute-based similarity matrix $\matS \approx \matY \Z^T$, its node embedding matrix $\matY$ can be approximated as 
	$$\tilde{\matY} = \matC \matU \matSigma^{1/2},$$ 
	where $\matC$ is the $n \times p$ matrix of similarities between the $n$ nodes and $p$ randomly chosen landmark nodes, and $\W^{\dagger}=\matU \matSigma \matV^\top$ is the full rank singular value decomposition of the pseudoinverse of the small $p \times p$ landmark-to-landmark similarity matrix $\matW$.
\end{theorem}

\begin{proof}
	Given the full-rank SVD of the $p \times p$ matrix $\W^\dagger$ as $\matU \matSigma \matV^\top,$ we can rewrite Eq.~\eqref{eq:approx} as $\matS \approx \tilde{\matS} = \matC (\matU \matSigma \matV^\top) \matC^\top = (\matC \matU \matSigma^{1/2}) \cdot (\matSigma^{1/2} \matV^\top \matC^\top) = \tilde{\Y}\tilde{\Z}^\top$. 
\end{proof}

Now, we \emph{never} have to construct an $n \times n$ matrix and then factorize it (i.e., by optimizing a nonconvex factorization objective).  
Instead, to derive $\tilde{\matY}$, the only node comparisons we need are for the $n \times p$ ``skinny'' matrix $\matC$, while the expensive SVD is performed only on its small submatrix $\W$.  Thus, we can obtain node representations by \emph{implicitly} factorizing $\tilde{\matS}$, a low-rank approximation of the full similarity matrix $\matS$. 
The $p$-dimensional node embeddings of the two input graphs $\graphG_1$ and $\graphG_2$ are then subsets of $\tilde{\matY}$: $\tilde{\matY}_1$ and $\tilde{\matY}_2$, respectively.  This construction corresponds to the explicit factorization (Fig.~\ref{nystrom-pipeline}), but at significant runtime and storage savings. 

\begin{figure}[t!]
    \centering
    \includegraphics[width=\columnwidth]{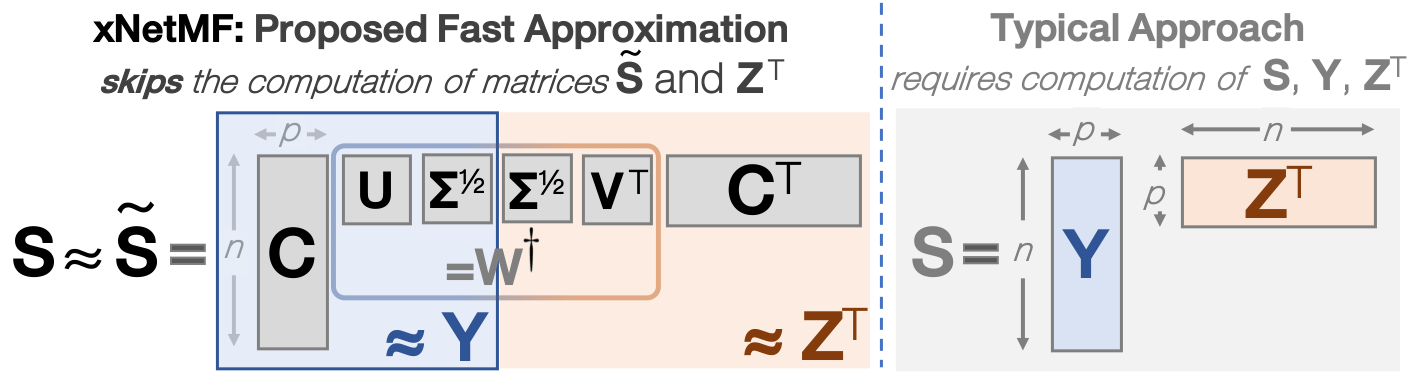}
    \vspace{-0.7cm}
    \caption{Proposed \embedding (using the SVD of $\matW^\dagger$) vs.\ typical matrix factorization for computing the node embeddings $\matY$. Our \embedding method leads to significant savings in space and runtime.}
    \label{nystrom-pipeline}
    \vspace{-.1cm}
\end{figure}

 As stated earlier, \embedding, which we summarize in Alg.~\ref{xnetmf}, forms the first two steps of \method.  The postprocessing step, where we normalize the magnitude of the embeddings, makes them more comparable based on Euclidean distance, which we use in \method.

\vspace{0.1cm}
\noindent \textbf{Connection between \embedding and SGNS}. 
We show a formal connection between matrix factorization, the technique behind our \embedding, and a variant of the struc2vec framework: another form of structure-based embedding optimized with SGNS~\cite{struc2vec} in Appendix \ref{app:xnetmfs2v}. 
Indeed, similar equivalences between SGNS and matrix factorization have been studied~\cite{levy2014neural,emf-embed} and applied to proximity-based node embedding methods \cite{netmf}, but
ours is the first to explore such connections for methods that preserve \textit{structural} identity. 

{\footnotesize
\begin{algorithm}[t!]
\caption{\method($\graphG_1, \graphG_2$, $p, K, \gamma_s, \gamma_a, \alpha$)} 
\label{regal}
\begin{algorithmic}[1]
\State \textbf{====== STEPS 1 and 2. Structural Node Representation Learning ====== }
\State $[\tilde{\matY}_1, \tilde{\matY}_2]$ = \embedding($\graphG_1, \graphG_2, p, K, \gamma_s, \gamma_a$) \Comment{Learn $n_1 \times p$ and $n_2 \times p$ embeddings} 
\vspace{0.1cm}
\State \textbf{=========== STEP 3. Fast Node Representation Alignment =========== }
\State $\matM$ = \textbf{empty} \Comment{\textbf{sparse} $n_1 \times n_2$ matrix $\matM$ of possible alignments}
\State $T$ = KDTree($\tilde{\matY}_2$) \Comment{Build a $k$-d tree on the node embeddings of $G_2$}
\vspace{0.1cm}
\State {/* \textbf{Match embeddings} to infer alignments */}
\For{$i = 1 \to n_1$} 
    \State {/* For embedding $i$ in $\graphG_1$, get the $\alpha$ most similar embed.\ in $\graphG_2$ and  distances*/}
    \State [TOP-$\alpha$, TOP-dist] = QueryKDTree(T, $\tilde{\matY}_1[i]$, $\alpha$) \Comment{$\tilde{\matY}_1[i]$: $i^{th}$ embedding }
  
    \For{$j$ in TOP-$\alpha$} 
        \State $m_{ij} = e^{-\text{TOP-dist}[j]}$ \Comment{Populating alignment matrix $\matM$ with embed.} 
        
    \EndFor \Comment{similarities: $e^{\text{TOP-dist}[j]}=e^{-||\,\tilde{\matY}_1[i]\,-\, \tilde{\matY}_2[j]\,||_2^2}$}
\EndFor

\State \Return{$\matM$} \Comment{alignments are largest entries in each row or column (Fig.~\ref{pipeline})}
\end{algorithmic}
\end{algorithm}
}

\setlength{\textfloatsep}{1\baselineskip plus 0.2\baselineskip minus 0.5\baselineskip}

\setlength{\textfloatsep}{0.3cm}
{\footnotesize
\begin{algorithm}[t!]
\caption{\embedding($\graphG_1, \graphG_2, p, K, \gamma_s, \gamma_a$)}
\label{xnetmf}
\begin{algorithmic}[1]
\State \textbf{================ STEP 1. Node Identity Extraction ================ }
\For{node $u$ in $\V_1 \cup \V_2$}  
    \For{hop $k$ up to $K$} \Comment{counts of node degrees of $k$-hop neighbors of $u$}
        \State $\vecdk$ = CountDegreeDistributions($\R$) \Comment{$1 \le K \le \text{graph diameter}$}
    \EndFor
    \State $\vecd_u =\sum_{k=1}^{K} \delta^{k-1} \vecdk$ \Comment{discount factor $\delta \in (0,1]$}
\EndFor
\vspace{0.1cm}
\State \textbf{========= STEP 2. Efficient Similarity-based Representation ========= }
\State \textbf{========== STEP 2a. Reduced $n \times p$ Similarity Computation ========== }
\State $\Ell = $ ChooseLandmarks($\graphG_1, \graphG_2$,$p$) \Comment{choose $p$ nodes from  $\graphG_1$, $\graphG_2$}
\For{node $u$ in $\V$}
    \For{node $v$ in $\Ell$} 
        \State $c_{uv} = e^{-\gamma_s \cdot ||\vecd_u - \vecd_v||_2^2 \,-\, \gamma_a \cdot \text{dist}(\vecf_u,\vecf_v)  }$
    \EndFor
\EndFor \Comment{Used in low-rank approx. of similarity graph (not constructed)}
\vspace{0.1cm}
\State \textbf{=========== STEP 2b. From Similarity to Representation ============ }
\State $\W = \matC[\Ell, \Ell]$ \Comment{Rows of $\matC$ corresponding to landmark nodes}
\State {$[\matU, \matSigma, \matV]$ = SVD($\W^\dagger$)} 

\State $\tilde{\matY} = \matC \matU \matSigma^{-\frac{1}{2}}$ \Comment{\textbf{Embedding}: \emph{implicit} factorization of similarity graph}

\State $\tilde{\matY} = Normalize(\tilde{\matY})$ \Comment{\textbf{Postprocessing}: make embeddings have magnitude 1} 
\State $\tilde{\matY}_1, \tilde{\matY}_2$ = Split($\tilde{\matY}$) \Comment{\textbf{Separate} representations for nodes in $\graphG_1$, $\graphG_2$}

\State \Return{$\tilde{\matY}_1, \tilde{\matY}_2$}
\end{algorithmic}
\end{algorithm}
}

\subsection{Step 3: Fast Node Representation Alignment} 
\label{sec:align} 

The final step of \method is to efficiently align nodes using their representations, assuming that two nodes  $u \in \V_1$ and $v \in \V_2$ may match if their \embedding embeddings are similar. 
Let $\tilde{\matY}_1$ and $\tilde{\matY}_2$ be matrices of the $p$-dimensional embeddings for nodes in graphs $\graphG_1$ and $\graphG_2$. 
We take the likeliness of (soft) alignment to be proportional to the similarity between the nodes' embeddings.  Thus, we greedily align nodes to their closest match in the other graph based on embedding similarity, as shown in Fig.~\ref{embed-pipeline}.  
This method is simpler and faster than optimization-based approaches, and works thanks to high-quality node feature representations.  

\vspace{0.1cm}
\noindent \textbf{Data structures for efficient alignment.}
A natural way to find the alignments for each node is to compute all pairs of similarities between node embeddings (i.e., the rows of $\tilde{\matY}_1$ and $\tilde{\matY}_2$) and choose the top-1 for each node.
Of course, this is not desirable due to its inefficiency.
Since in practice only the top-$\alpha$ most likely alignments are used, we turn to specialized data structures for quickly finding the closest data points.   
We store the embeddings $\tilde{\matY}_2$ in a $k$-d tree, a data structure used to accelerate exact similarity search for nearest neighbor algorithms and many other applications~\cite{kdtree}.  

For each node in $\graphG_1$, we can quickly query this tree with its embedding to find the $\alpha << n$ closest embeddings from nodes in $\graphG_2$.  
This allows us to compute ``soft'' alignments for each node  by returning one or more nodes in the opposite graph with the most similar embeddings, unlike many existing alignment methods that only find ``hard'' alignments~\cite{netalign, final, isorank, klau}.
Here, we define the similarity between the $p$-dimensional embeddings of nodes $u$ and $v$ as $sim_{emb}(\tilde{\matY}_1[u], \tilde{\matY}_2[v]) = e^{-||\,\tilde{\matY}_1[u]\, - \,\tilde{\matY}_2[v]\,||_2^2}$, which converts the Euclidean distance to similarity. Since we only want to align nodes to counterparts in the other graph,  we only compare embeddings in $\tilde{\matY}_1$ with ones in $\tilde{\matY}_2$.   If multiple top alignments are desired, they may be returned in \textbf{sorted} order by their embedding similarity; we use sparse matrix notation in the pseudocode just for simplicity.

\subsection{Complexity Analysis}

Here we analyze the computational complexity of each step of \method. 
To simplify notation, we assume both graphs have $n_1=n_2=n'$ nodes. 
    
\begin{enumerate}
    \item \textbf{Extracting node identity}: It takes approximately $O(n'K d_{avg}^2)$ time, finding neighborhoods up to hop distance $K$ by joining the neighborhoods of neighbors at the previous hop: formally, we can construct $\R = \bigcup_{v \in \mathcal{R}_u^{k-1} } \mathcal{R}_v^1 - \bigcup_{i=1}^{k-1} \mathcal{R}_u^{i} $.  We could also use breadth-first search from each node to compute the $k$-hop neighborhoods in $O(n'^3)$ worst case time---in practice significantly lower for sparse graphs and/or small $K$---but we find that this construction is faster in practice.   
    
    \item \textbf{Computing similarities}: We compute the similarities of the length-$b$ features (weighted counts of node degrees in the k-hop neighborhoods, split into $b$ buckets) between each node and $p$ landmark nodes: this takes $O(n'pb)$ time.
    
    \item \textbf{Obtaining representations}: We first compute the pseudoinverse and SVD of the $p \times p$ matrix $\matW$ in time $O(p^3)$, and then left multiply it by $\matC$ in time $O(n'p^2)$.  Since $p << n'$, the total time complexity for this step is $O(n'p^2)$.
    
    \item \textbf{Aligning embeddings}: We construct a $k$-d tree and use it to find the top alignment(s) in $\graphG_2$ for each of the $n'$ nodes in $\graphG_1$ in average-case time complexity $O(n' \log n')$.  
\end{enumerate}
    
The total complexity is $O(n' \max\{pb, p^2, Kd_{avg}^2, \log n'\})$.  As we show experimentally, it suffices to choose small $K$ as well as $p$ and $b$ logarithmic in $n'$.
With $d_{avg}$ often being small in practice,   
this can yield \emph{sub-quadratic} time complexity.  
It is straightforward to show that the space requirements are sub-quadratic as well.

\section{Experiments}
\label{sec:experiments}
We answer three important questions about our methods: \\
\textbf{(Q1)} How does \method compare to baseline methods for network alignment on noisy real world datasets (Table \ref{tab:datasets}), with and without attribute information, in terms of accuracy and runtime? \\
\textbf{(Q2)}~How scalable is \method? \\ 
\textbf{(Q3)} How sensitive are \method and \embedding to hyperparameters?

\vspace{0.2cm}
\noindent \textbf{Experimental Setup.} 
Following the network alignment literature~\cite{bigalign,final}, for each real network dataset with \emph{adjacency matrix} $\A$, we generate a new network with adjacency matrix $\A' = \perm\A\perm^\top$, where 
$\perm$ is a randomly generated permutation matrix with the nonzero entries representing ground-truth alignments. 
We add structural noise to $\A'$ by removing edges with probability $p_s$ without disconnecting any nodes.  

For experiments with attributes, we generate synthetic attributes for each node if the graph does not have any.
We add noise to these by flipping binary values or choosing categorical attribute values uniformly at random from the remaining possible values with probability $p_a$.  
For each dataset and noise level, noise is randomly and independently added.   

All experiments are performed on an Intel(R) Xeon(R) CPU E5-1650 at 3.50GHz with 256GB RAM, with hyperparameters  $\delta = 0.01$, $K = 2$, $\gamma_{s} = \gamma_{a} = 1$, and $p = \lfloor 10 \log_2 n \rfloor$ unless otherwise stated. 
Landmarks for \method are chosen arbitrarily from among the nodes in our graphs, in keeping with the effectiveness and popularity of sampling uniformly at random \cite{nystrom}.  
In Sec.~\ref{sec:parameters}, we explore the parameter choices
and find that these settings yield stable results at reasonable computational cost. 

\begin{table*}[htb]
  \centering
  \begin{minipage}[c]{0.63\textwidth}
  \includegraphics[width=\textwidth]{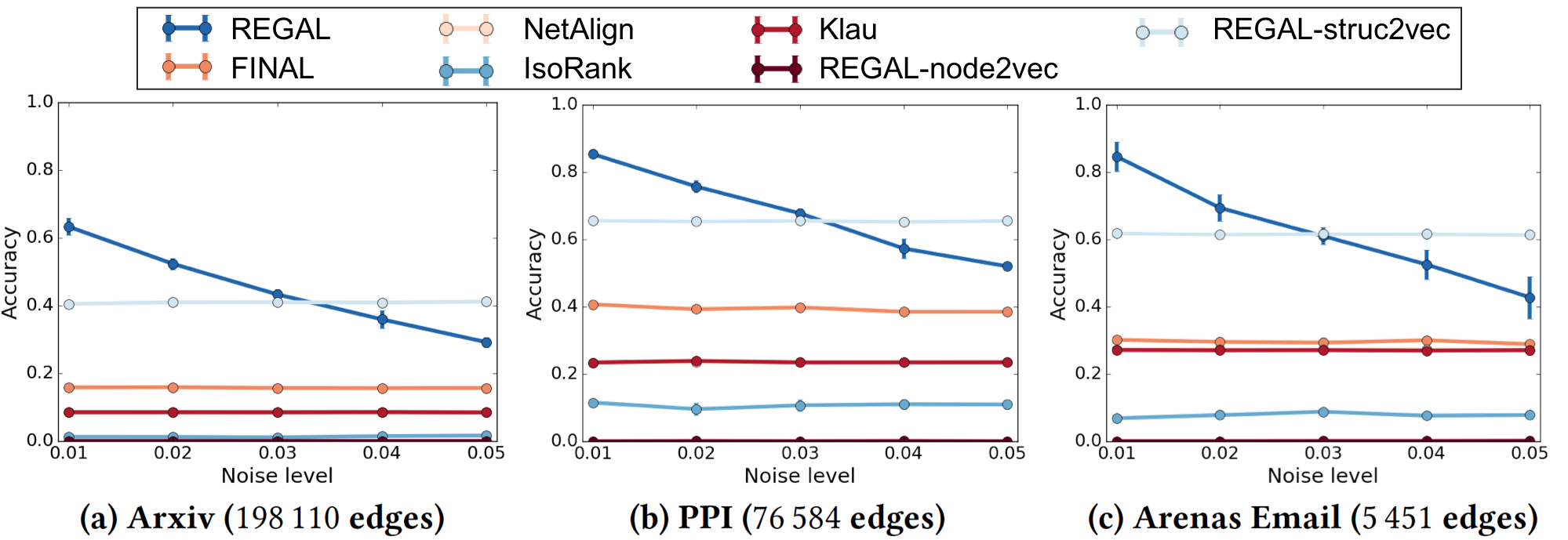}
  \vspace{-0.6cm}
   \captionof{figure}{Accuracy of network alignment methods with varying $p_s$. \method (in dark blue) achieves consistently high accuracy \emph{and} runs faster than its closest competitors (Table~\ref{structural-comp-runtime}).
    }
   \label{structural-comp}
  \end{minipage}
  \hfill
  \begin{minipage}[c]{.34\textwidth}
  \caption{Average (stdev) runtime in sec of alignment methods from 5 trials. The two fastest methods per dataset are in bold. \method is faster than its closest competitors in accuracy (Fig.~\ref{structural-comp}).}
  \label{structural-comp-runtime}
  \resizebox{\columnwidth}{!}{
  \begin{tabular}[b]{l rrr }
    \toprule
   \textbf{Dataset}  &  \textbf{Arxiv} &  \textbf{PPI} & \textbf{Arenas} \\
     \midrule
     \textbf{FINAL}  & 4182 (180) & 62.88 (32.20)  & 3.82 (1.41)  \\
     \textbf{NetAlign}  & 149.62 (282.03) & 22.44 (0.61) & \bf 1.89 (0.07) \\ 
     \textbf{IsoRank}  & {\bf 17.04 (6.22)} & {\bf 6.14 (1.33)} & {\bf 0.73 (0.05)} \\
     \textbf{Klau}   & 1291.00 (373) & 476.54 (8.98) & 43.04 (0.80) \\ \hline
     \textbf{REGAL-node2vec}  & 709.04 (20.98) & 139.56 (1.54)  & 15.05 (0.23)   \\
     \textbf{REGAL-struc2vec}  & 1975.37 (223.22) & 441.35 (13.21) & 74.07 (0.95)   \\
     \textbf{REGAL}  & {\bf 86.80 (11.23)} & {\bf 18.27 (2.12)} & {2.32 (0.31)}   \\
    \bottomrule
    \end{tabular}
    }
  \end{minipage}
\end{table*}

\vspace{0.15cm}
\noindent \textbf{Baselines.} We compare against six baselines.
Four are well known existing network alignment methods and two are variants of our proposed framework that match embeddings produced by existing node embedding methods (i.e.,  not \embedding).
The \textbf{four existing network alignment methods} are: \textbf{(1) FINAL}, which introduces a family of algorithms optimizing quadratic objective functions~\cite{final}; \textbf{(2) NetAlign}, which formulates alignment as an integer quadratic programming problem and solves it with message passing algorithms~\cite{netalign}; \textbf{(3) IsoRank}, which solves a version of the integer quadratic program with relaxed constraints~\cite{isorank}; and \textbf{(4) Klau's} algorithm (Klau), which imposes a linear programming relaxation, decomposes the symmetric constraints and solves it iteratively~\cite{klau}. 
These methods all require as input a matrix containing prior alignment information, which we construct from degree similarity, taking the top $\lfloor \log_2 n \rfloor$ entries for each node; \method, by contrast, does not require prior alignment information.

For the \textbf{two variants} of our framework, which we refer to as \textbf{(5) \method-node2vec} and (6) \textbf{\method-struc2vec}, we replace our own \embedding embedding step (i.e., Steps 1 and 2 in \method) with existing node representation learning methods  node2vec~\cite{node2vec} or  struc2vec~\cite{struc2vec}: 
two recent, state-of-the-art node embedding methods that make a claim about being able to capture some form of structural equivalence.  To apply these embedding methods, which were formulated for a single network,  we create a single  input graph $\graphG$ by combining the graphs with respective adjacency matrices $\matA$ and $\matA'$ into one block-diagonal adjacency matrix $[\matA \; \mathbf{0}; \mathbf{0} \; \matA']$. 
Beyond the input, we use their default parameters: 10 random walks of length 80 for each node to sample context with a window size of 10. For node2vec, we set $p = q = 1$ (other values make little difference).  For struc2vec, we use the recommended optimizations~\cite{struc2vec} to compress the degree sequences and reduce the number of node comparisons, which were found to speed up computation with little effect on performance \cite{struc2vec}. As we do for our \embedding method, we consider a maximum hop distance of $K = 2$.  

\vspace{0.15cm}
\noindent \textbf{Metrics.} We compare \method to baselines with two metrics: 
\textbf{alignment accuracy}, which we take as
(\# correct alignments) / (total \# alignments), and \textbf{runtime}.
When computing results, we average over 5 independent trials on each dataset at each setting (with different random permutations and noise additions) and report the mean result and the standard deviation (as bars around each point in our plots.)  
We also show where \method's soft alignments contain the ``correct'' similarities within its top $\alpha << n$ choices using the more \textbf{general top-$\alpha$ accuracy}: (\# correct alignments in top-$\alpha$ choices) / (total \# alignments).  
This metric does not apply to the existing network alignment baselines that do not directly match node embeddings and only find hard alignments.

\begin{table}[t!]
\vspace{-0.1cm}
\centering
\caption{Real data used in our experiments.}
\label{tab:datasets}
\vspace{-0.3cm}
{\footnotesize
\begin{tabular}{l r@{\hspace{1.2em}}  r@{\hspace{4pt}} @{\hspace{5pt}}r l}
\toprule
   \textbf{Name} & \textbf{Nodes} & \textbf{Edges} &  \textbf{Description}  \\
\midrule
     Facebook \cite{fblarge} & 63\,731 & 817\,090 & social network \\ 
     
     Arxiv \cite{snapnets} & 18\,722 & 198\,110 & collaboration network \\ 

     DBLP \cite{prado2013mining} & 9\,143 & 16\,338 & collaboration network \\

     PPI \cite{ppi} & 3\,890     & 76\,584    &  protein-protein interaction \\

     Arenas Email \cite{koblenz} &  1\,133      & 5\,451    & communication network \\
\bottomrule
\end{tabular}
}
\vspace{-0.1cm}
\end{table}

  \begin{figure*}[t!]
\vspace{-0.6cm}
	\centering
    \subfloat[1 synthetic binary attribute\label{1synth}]{%
      \includegraphics[width=0.2\textwidth]{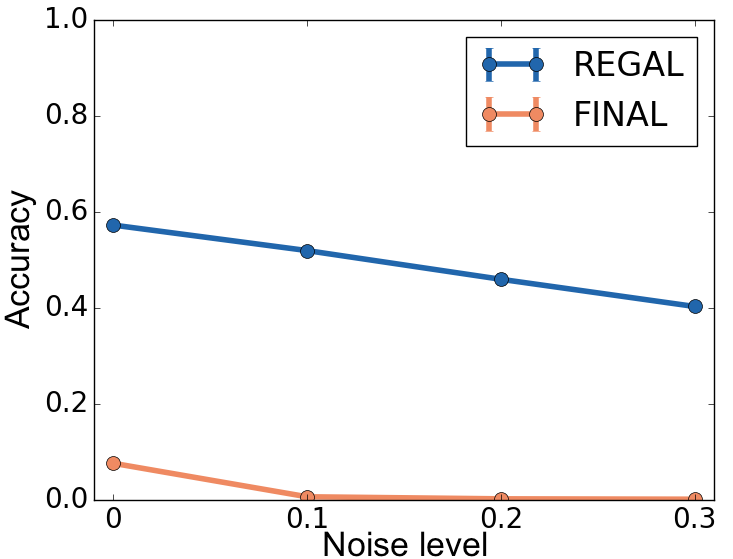}
    }
    \subfloat[3 synthetic binary attributes\label{3synth}]{%
      \includegraphics[width=0.2\textwidth]{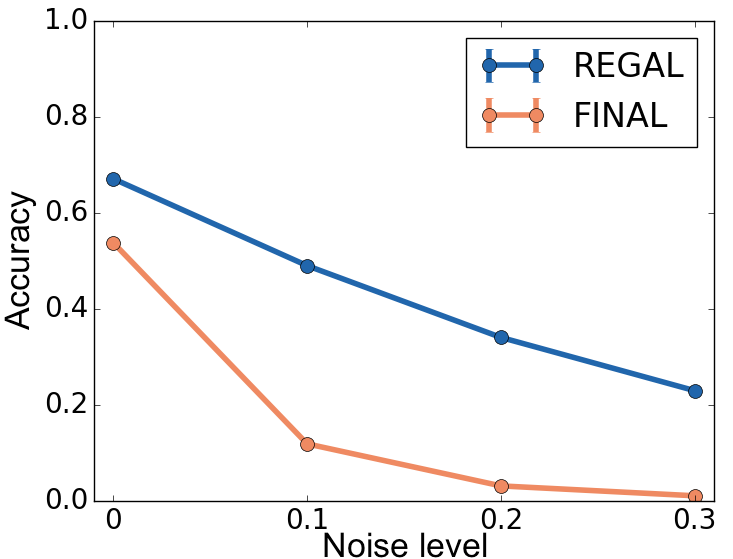}
    }
    \subfloat[5 synthetic binary attributes\label{5synth}]{%
      \includegraphics[width=0.2\textwidth]{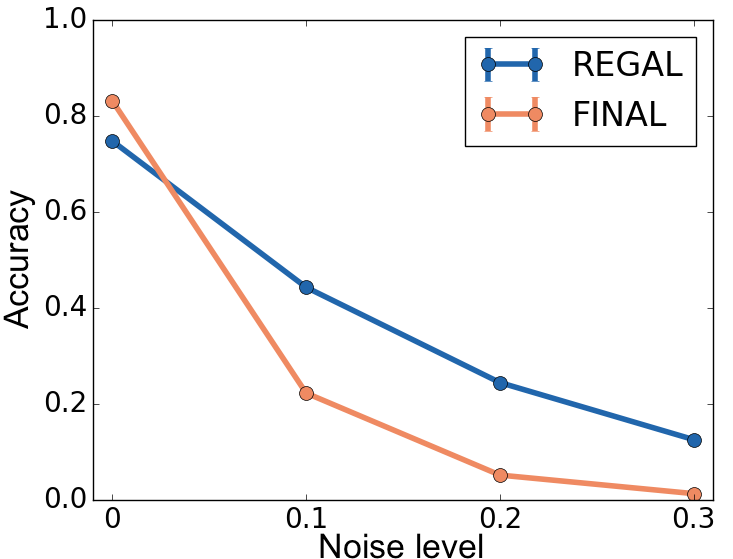}
    }
    \subfloat[Real attribute (29 values)
    \label{realattrs}]{\includegraphics[width=0.2\textwidth]{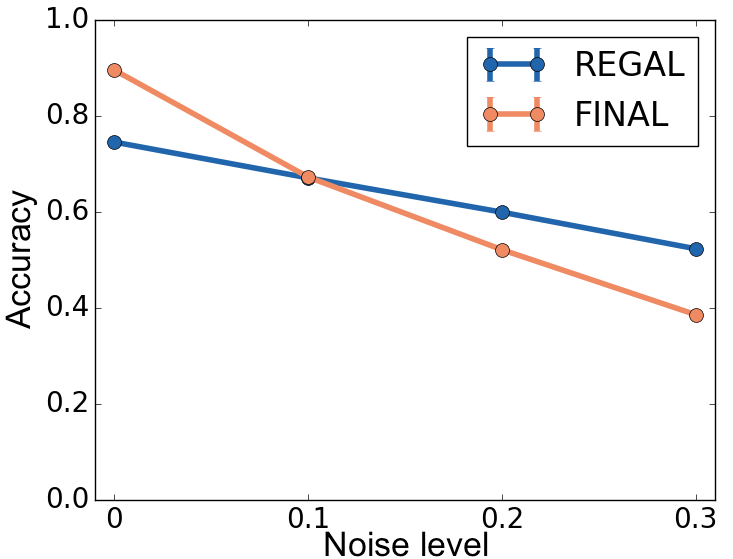}
    }
    \subfloat[Runtime with attributes
    \label{attr-run}]{\includegraphics[width=0.22\textwidth]{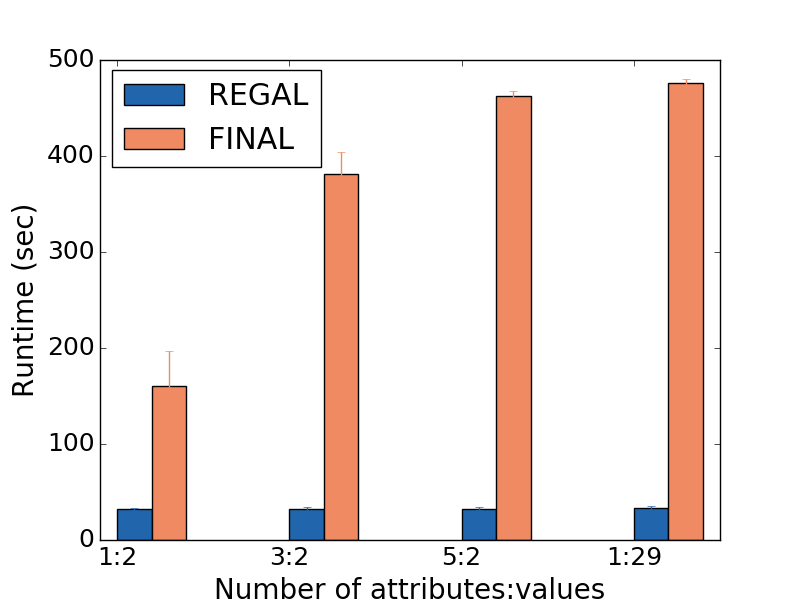}
    }
    \vspace{-0.2cm}
	\caption{DBLP Network alignment with varying $p_a$: \method is more robust to attribute noise (plots a-d) and runs faster (plot e) than FINAL for various numbers and types of attributes. In (e) the x axis consists of $<$\# of attributes: \# of values$>$ pairs corresponding to plots (a)-(d).}
    \label{attrs-comp}
    \vspace{-0.1cm}
\end{figure*}

\subsection{Q1: Comparative Alignment Performance}  
\label{sec:exp_q1}

To assess the comparative performance of \method versus existing network alignment methods on a variety of challenging datasets, we perform two experiments studying the effects of structural and attribute noise, respectively.

\subsubsection{Effects of structural noise.} 
In this experiment we study how well \method matches nodes based on structural identity alone. 
This also allows us to compare to the baseline network alignment methods NetAlign, IsoRank, and Klau, as well as the node embedding methods node2vec and struc2vec, none of which was formulated to handle or align attributed graphs (which we study in Sec.~\ref{sec:attr_noise}).
As we discuss further below, \method is one of the fastest network alignment methods, especially on large datasets, and has comparable or better accuracy than all baselines.

\vspace{.1cm}
\noindent \textbf{Results.} \textbf{(1) Accuracy.} The accuracy results on several datasets are shown in Figure~\ref{structural-comp}.  
The \emph{structural} embedding \method variants consistently perform best. 
Both \method (matching our proposed \embedding embeddings) and \method-struc2vec are significantly more accurate than all \textit{non-representation learning baselines} across noise levels and datasets.  
As expected, \method-node2vec does hardly better than random chance because rather than preserving structural similarity, it preserves similarity to nodes based on their proximity to each other, which means there is no way of identifying similarity to corresponding nodes in \emph{other}, disconnected graphs (even when we combine them into one large graph, because they form disconnected components.)  This major limitation of embedding methods that use proximity-based node similarity criteria \cite{mlg_paper} justifies the need for \emph{structural} embeddings for cross-network analysis.

Between \method and \method-struc2vec, the two highest performers, \method performs better with lower amounts of noise. 
This is likely because struc2vec's randomized context sampling introduces some variance into the representations that \embedding does not have, as nodes that should match will have different embeddings not only because of noise, but also because they had different contexts sampled.  
With higher amounts of noise (4-5\%), \method outperforms \method-struc2vec in speed,  but at the cost of some accuracy. It is also worth noting that their accuracy margin is smaller for larger graphs. 
On larger datasets, our simple and fast logarithmic binning scheme (Step 1 in Sec.~\ref{sec:step1}) provides a robust enough way of comparing nodes with high expected degrees.  However, on small graphs with a few thousand nodes and edges, it appears that struc2vec's use of dynamic time warping (DTW) better handles misalignment of degree sequences from noise because it is a nonlinear alignment scheme.  
Still, we will see that \method is significantly faster than its struc2vec variant, since DTW is computationally expensive \cite{struc2vec}, as is context sampling and SGNS training.   

\vspace{0.1cm}
\noindent \textbf{(2) Runtime.}
In Table~\ref{structural-comp-runtime}, we  compare the average runtimes of all different methods across noise levels. We observe that \method scales significantly better using \embedding than when using other node embedding methods. Notably, \method is 6-8$\times$ faster than \method-node2vec and 22-31$\times$ faster than \method-struc2vec.
This is expected as both dynamic time warping (in struc2vec) and context sampling for SGNS (in struc2vec and node2vec) come with large computational costs.
\method, at the cost of some robustness to high levels of noise, avoids both the variance and computational expense of random-walk-based sampling.  This is a significant benefit that allows \method to achieve up to an order of magnitude speedup over the other node embedding methods.  Additionally, \method is able to leverage the power of node representations and also use attributes, unlike the other representation learning methods.

Comparing to baselines that do not use representation learning, we see that \method is  competitive in terms of runtime as well as significantly more accurate.  
\method is consistently faster than FINAL and Klau, the next two best-performing methods by accuracy (NetAlign is virtually tied for third place with Klau on all datasets).  
Although NetAlign runs faster than \method on small datasets like Arenas, on larger datasets like Arxiv NetAlign's message passing becomes expensive.  
Finally, while IsoRank is consistently the fastest method, it performs among the worst on all datasets in accuracy.  
Thus, we can see that our \method framework is also one of the fastest network alignment methods as well as the most accurate.  

\subsubsection{Effects of attribute-based noise.} 
\label{sec:attr_noise} In the second experiment, we study \method's comparative sensitivity to $p_a$ when we use node attributes.  
Here we compare \method to FINAL because it is the only baseline that handles attributes. 
We also omit embedding methods othen than \embedding, since they operate on plain graphs.

We study a subnetwork of a larger DBLP collaboration network extracted in \cite{final} (Table \ref{tab:datasets}).
This dataset has 1 node attribute with 29 values, corresponding to the top conference in which each author (a node in the network) published.  
This single attribute is quite discriminatory: with so many possible attribute values, a comparatively smaller number of nodes share the same value. 
We add $p_s = 0.01$ structural noise to randomly generated permutations.   

We also  increase  attribute information by increasing the number of attributes.  
To do so, we simulate different numbers of binary attributes.   We study somewhat higher levels of attribute noise, as they are not strictly required for network alignment.   

\vspace{.1cm}
\noindent \textbf{Results.}  In Figure \ref{attrs-comp}, we see that \method mostly outperforms FINAL in the presence of attribute noise (both for real and multiple synthetic attributes), \emph{or} in the case of limited attribute information (e.g., only 1-3 binary attributes in Fig.~\ref{1synth}-\ref{5synth}). This is because FINAL relies heavily on attributes, whereas \method uses structural and attribute information in a more balanced fashion.

While FINAL achieves slightly higher accuracy than \method with abundant attribute information from many attributes or attribute values and minimal noise (e.g. the real attribute with 29 values in Figure \ref{realattrs}, or 5 binary attributes in Figure \ref{5synth}), this is expected due to FINAL's reliance on attributes. 
Also, in Figure~\ref{attr-run} where we plot the runtime with respect to number of $<$attributes : attribute values$>$, we see FINAL incurs significant runtime increases as it uses extra attribute information. 
Even without these added attributes, \method is up to two orders of magnitude faster than FINAL.

\begin{figure*}[t!]
	\centering
    \subfloat[Discount factor $\delta$
    \label{delta}]{\includegraphics[width=0.2\textwidth]{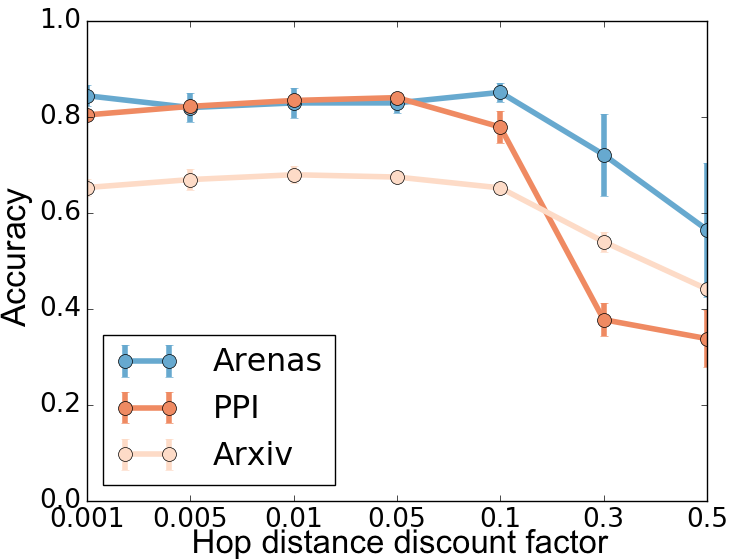}
    }
    \subfloat[Maximum hop distance $K$ \label{maxlayer}]{
      \includegraphics[width=0.2\textwidth]{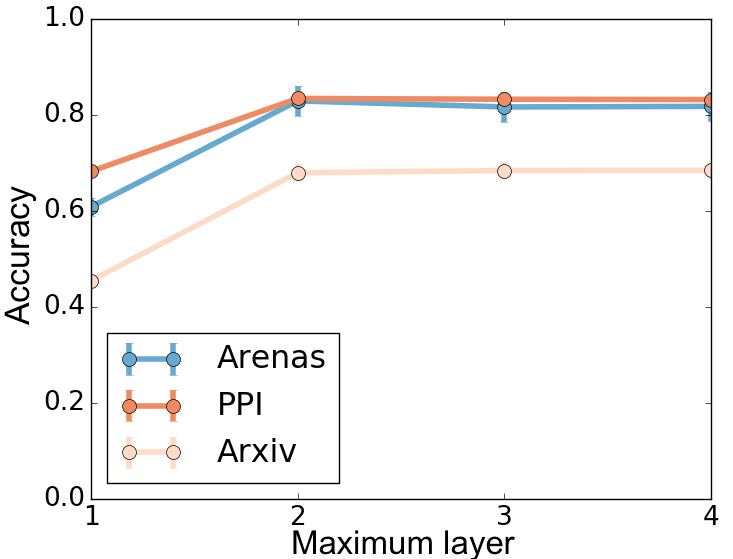}
    }
    \subfloat[Coeff. $\gamma_s$ (structural sim.)\label{gammastruc}]{
      \includegraphics[width=0.2\textwidth]{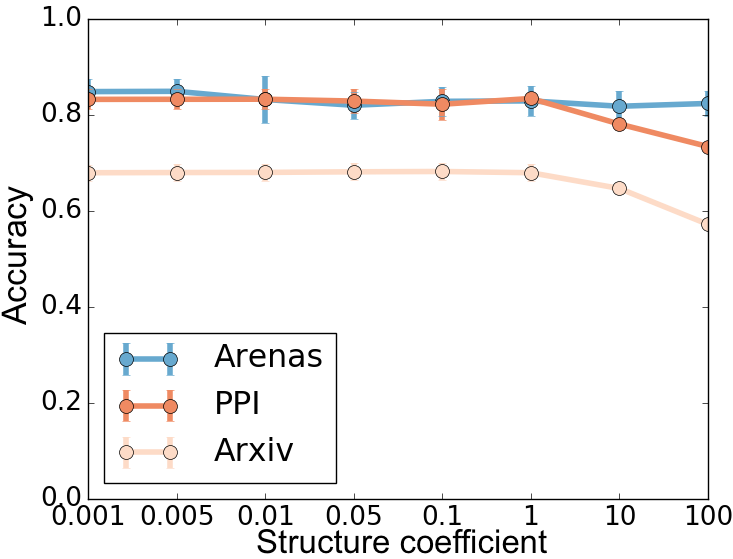}
    }
    \subfloat[Coeff. $\gamma_a$ (attribute sim.)\label{gammaattr}]{
      \includegraphics[width=0.2\textwidth]{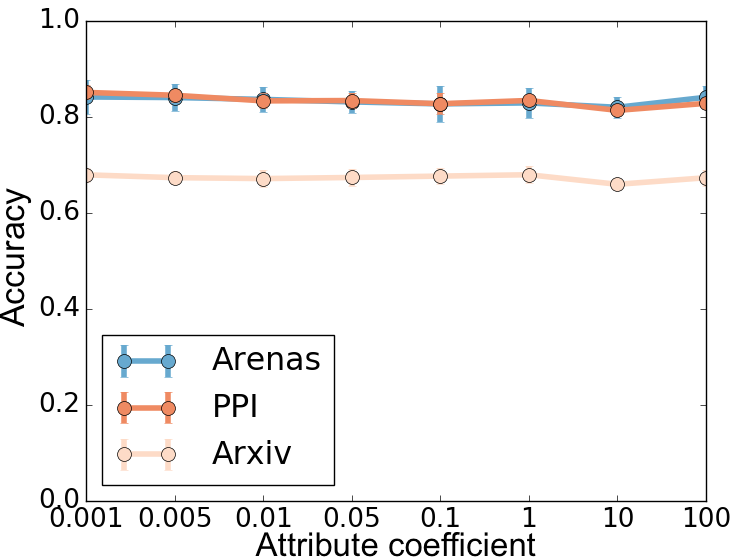}
    }
    \subfloat[top-$\alpha$ scores on Facebook~\cite{fblarge} \label{topa}]{%
      \includegraphics[width=0.2\textwidth]{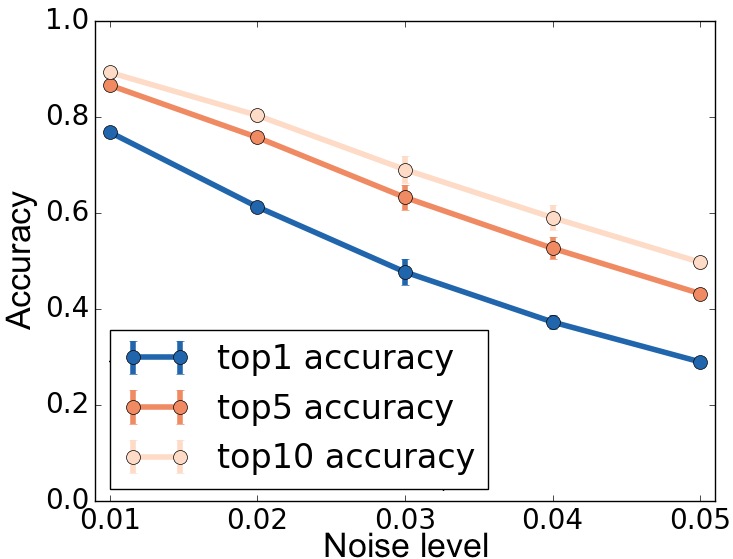}
    }
    \vspace{-0.35cm}
	\caption{Robustness of \method to hyperparameters on different datasets: \method is generally robust for a range of values, without fine tuning.}
    \label{params}
    \vspace{-0.2cm}
\end{figure*}

\subsection{Q2: Scalability}
\label{sec:scalability}

To analyze the scalability of REGAL, we generate Erd\"{o}s-R\'{e}nyi graphs with $n = 100$ to 1,000,000 nodes and constant average degree 10, along with one binary attribute.  We generate a randomized, noisy permutation ($p_s = 0.01, p_a = 0.05$) and look for the top $\alpha = 1$ alignments.  Thus, we embed \emph{both} graphs--double the number of nodes in a single graph.  Figure \ref{scalability} shows the runtimes for the major steps of our methods.

\begin{wrapfigure}{r}{0.22\textwidth}
  \vspace{-0.95cm}
    \centering
    \includegraphics[width=0.22\textwidth]{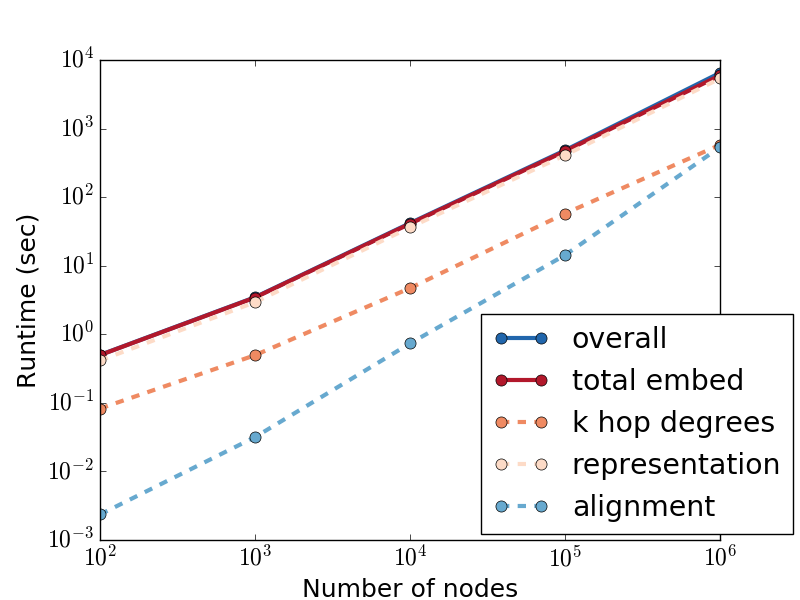}
    \vspace{-0.7cm}
    \caption{\method is subquadratic.}
    \label{scalability}
    \vspace{-0.4cm}
    \end{wrapfigure}
    
\textbf{Results.}  We see that the total runtimes of \method's steps 
are clearly \emph{sub-quadratic}, which is rare for alignment tasks.  
In practice this means that \method can scale to very large networks.  
The dominant step is computing $O(n \log n)$ similarities to landmarks in $\matC$ and using this to form the Nystr\"{o}m-based representation.
The alignment time complexity grows the most steeply, as the dimensionality $p$ grows with the network size and increasingly affects lookup times.  
In practice, though, the alignment adds little overhead time, even for the largest graph, because of the $k$-d tree.  
\emph{Without} it, \method runs out of memory on 100K or more nodes.   

From a practical perspective, while our current implementation is single-threaded, many steps---including the expensive embedding construction and alignment steps---are easily and trivially parallelizable, offering possibilities for even greater speedups. 

\subsection{Q3: Sensitivity Analysis}
\label{sec:parameters}
To understand how \method's hyperparameters affect performance, we analyze accuracy by varying hyperparameters in several experiments.
For brevity, we report results at $p_s = 0.01$ and with a single binary noiseless attribute, although further experiments with different settings yielded similar results.  
Overall we find that \method is \textbf{robust} to different settings and datasets, indicating that \method can be applied readily to different graphs without requiring excessive domain knowledge or fine-tuning.

\vspace{.1cm}
\noindent \textbf{Results.} \textbf{(1) Discount factor $\delta$ and max hop distance $K$.} Figures \ref{delta} and \ref{maxlayer} respectively show the performance of \method as a function of $\delta$, the discount factor on further hop distances, and $K$, the maximum hop distance to consider.  We find that some higher-order structural information does help (thus $K = 2$ performs slightly better than $K = 1$), but only up to a point.
Beyond approximately 2 layers out, the structural similarity is so tenuous that it primarily adds noise to the neighborhood degree distribution (furthermore, computing further hop distances adds computational expense).   
Choosing $\delta$ between $0.01$--$0.1$ tends to yield best  performance. 
Larger discount factors $\delta$ tend to do poorly, though extremely small values may lose higher-order structural information.    

\begin{figure}[t!]
\vspace{-0.4cm}
	\centering
    \subfloat[Accuracy w.r.t. \# of landmarks\label{pacc}]{
      \includegraphics[width=0.25\textwidth]{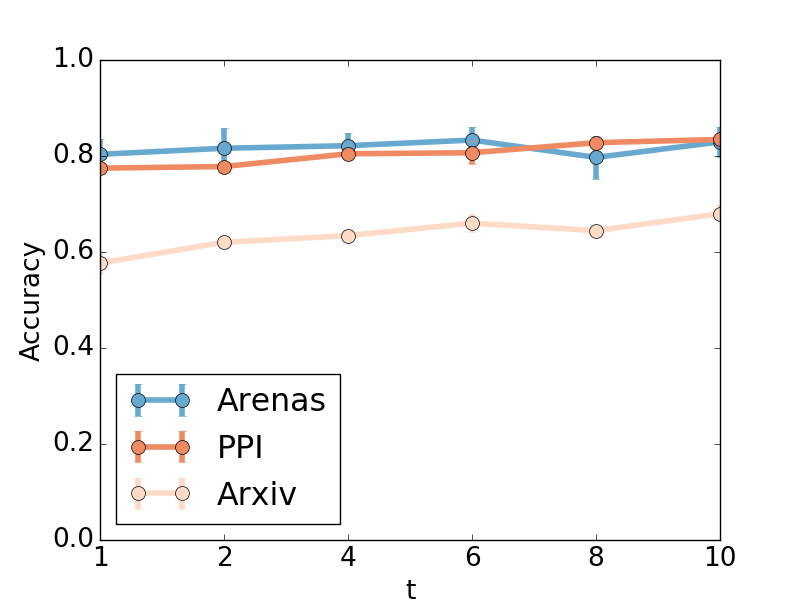}
    }
    ~
    \subfloat[Runtime w.r.t. \# of landmarks \label{prun}]{
      \includegraphics[width=0.25\textwidth]{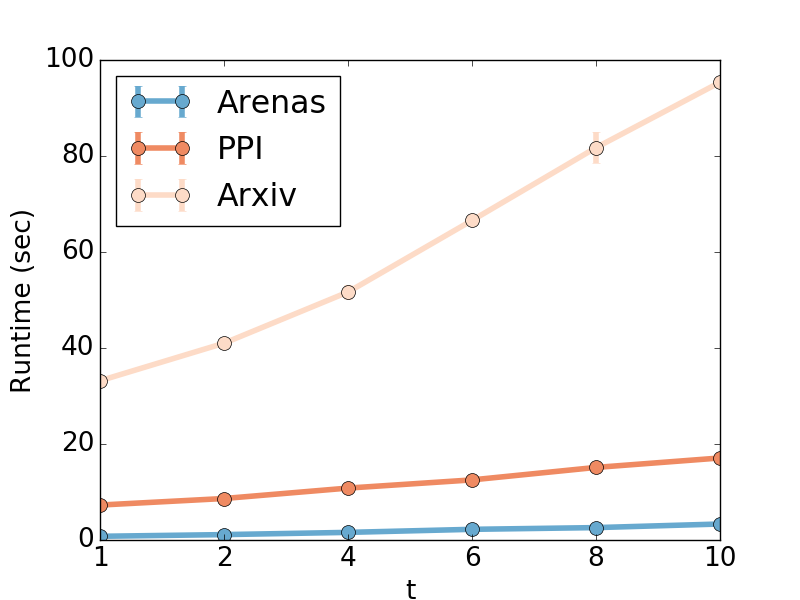}
    }
    \vspace{-0.2cm}
	\caption{Robustness of \method to $t$, which controls the number of landmarks $p = \lfloor t \log_2n \rfloor$: choosing more landmarks is more computationally expensive but can slightly increase accuracy.}
    \label{landmarks}
\end{figure}

\vspace{.05cm}
\noindent  \textbf{(2) Weights of structural $\gamma_{s}$ and attributed $\gamma_{a}$ similarity.} Next, we explore how to set the coefficients on the terms in the similarity function weighting structural and attribute similarity, which also governs a tradeoff between structural and attribute identity.  In Figs.~\ref{gammastruc} and \ref{gammaattr} we respectively vary $\gamma_{s}$ and $\gamma_{a}$ while setting the other to be 1. In general, setting these parameters to be 1, our recommended default value,  does fairly well.  Significantly larger values yield less stable performance.  

\vspace{.05cm}
\noindent \textbf{(3) Dimensionality of embeddings $p$.} To study the effects of the rank of the implicit low-rank approximation, which is also the dimensionality of the embeddings, we set the number of landmarks $p$ equal to $\lfloor t \log_2n \rfloor$ and vary $t$.  
Figure~\ref{pacc} shows that the accuracy is generally highest for the highest values of $t$, but Figure~\ref{prun} shows the expected increase in \method's runtime as more similarities are computed in $\matC$ and higher-dimensional embeddings are compared.  
To spare no expense in maximizing accuracy we use $t = 10$.
However, fewer landmarks still yield almost as high accuracy if computational constraints or high dimensionality are issues.

\vspace{.05cm}
\noindent \textbf{(4) Top-$\alpha$ accuracy.} It is worth studying not just the proportion of correct hard alignments, but also the top-$\alpha$ scores of the soft alignments that \method can return.  We perform alignment without attributes on a large Facebook subnetwork \cite{fblarge} and visualize the top-1, top-5, and top-10 scores in Fig.~\ref{topa}.  Across noise settings, the top-$\alpha$ scores are considerably several percentage points higher than the top-1 scores, indicating that even when \method misaligns a node, it often still recognizes the similarity of its true counterpart.   
\method's ability to find soft alignments could be valuable in many applications, like entity resolution across social networks \cite{bigalign}. 

\section{Conclusion}
\label{conclusion}
Motivated by the numerous applications of network alignment in social, natural, and other sciences, we proposed \method, a network alignment framework that leverages the power of node representation learning by aligning nodes via their learned embeddings.
To efficiently learn node embeddings that are comparable across multiple networks, we introduced \embedding within \method.
To the best of our knowledge, we are the first to propose an unsupervised representation learning-based network alignment method.

Our embedding formulation captures node similarities using structural and attribute identity, making it suitable for cross-network analysis. 
Unlike other embedding methods that sample node context with computationally expensive and variance-inducing random walks, 
our extension of the \nystrom low-rank approximation  allows us to implicitly factorize a similarity matrix \textit{without} having to fully construct it.
Furthermore, we showed that our formulation is a matrix factorization perspective on the skip-gram objective optimized over node context sampled from a similarity graph.  
Experimental results showed that \method is up to 30\% more accurate than baselines and $30\times$ faster in the representation learning stage.  Future directions include extending our techniques to weighted networks and incorporating edge signs or other attributes. 

\section*{Acknowledgements}
{\small 
This material is based upon work supported by the National Science Foundation under Grant No. IIS 1743088, an Adobe Digital Experience research faculty award, and the University of Michigan. Any opinions, findings, and conclusions or recommendations expressed in this material are those of the author(s) and do not necessarily reflect the views of the National Science Foundation or other funding parties. The U.S. Government is authorized to reproduce and distribute reprints for Government purposes notwithstanding any copyright notation here on.
}

\balance 

\bibliographystyle{ACM-Reference-Format}
\bibliography{sample}

\appendix

\section{Connections: \lowercase{x}N\lowercase{et}MF and SGNS}
\label{app:xnetmfs2v}

Here we unpack the key components of the struc2vec framework~\cite{struc2vec}, a random walk-based structural representation learning approach, and we find a matrix factorization interpretation at the heart of it. 

Given a (single-layer) similarity graph $\matS$, for each node $v$, struc2vec samples context nodes $\mathcal{C}$ with $m$ random walks of length $\ell$ starting from $v$. The probability of going from node $u$ to node $v$ is proportional to the nodes' (structural) similarity $s_{uv}$. This yields a co-occurrence matrix $\matD$: $d_{uv} = \#(u, v)$ is the number of times node $v$ was visited in context of node $u$.  Afterward, struc2vec optimizes a skip-gram objective function with negative sampling (SGNS):  

\begin{equation}
\max_{\matY,\matC} \sum_{y \in \V, c \in \mathcal{C}} \#(y,c) \log \sigma(\y^\top \context) + \ell \cdot  \mathbb{E}_{\context' \sim P_D} \log \sigma(-\y^\top \context') 
\label{sgns}
\end{equation}
where $\y$ and $\context$ are the embeddings of a node $y$, and its context node $c$, resp.;  $P_D(c) =  \sum_{y \in \V} \#(y,c) / \sum_{y \in \V, c \in C} \#(y,c)$ 
is the empirical probability that a node is sampled as some other node's context; and $\sigma(x) = (1 + e^{-x})^{-1}$ is the sigmoid function.  Analysis of SGNS for word embeddings \cite{emf-embed} showed under some assumptions on the upper bound of the co-occurrence count between two words that the objective of SGNS in Eq.~\eqref{sgns} is equivalent to matrix factorization of the co-occurrence matrix $\matD$, or MF$(\matD, \matY^\top \matC)$. Here MF is the objective of matrix factorization on $\matD$ (formally defined in \cite{emf-embed}, but in practice other matrix factorization techniques work well).  

Now, under these assumptions, we show a connection between optimizing Eq.~\eqref{sgns} with context sampled from the similarity graph (as in struc2vec), and factorizing the graph (as in \embedding).  

\begin{lemma}
\label{sgns-mf}
Equation~\eqref{sgns}, defined over a context sampled by performing $m$ length-1 random walks per node over $\matS$, is equivalent to \emph{MF}$(\matS, \matY^\top \matC)$ in the limit as $m$ goes to $\infty$, up to scaling of $\matS$.
\end{lemma}

\begin{proof}
This follows from the Law of Large Numbers.  As $m \rightarrow \infty$, the co-occurrence matrix $\matD$ converges to its expectation.  This is just $m\cdot \matS$, since $d_{ij}$ is the \# of times node $v_j$ is sampled in a random walk of length 1 from $v_i$, which is equal to the \# of walks from node $v_i$ times the probability that the walk goes to $v_j$ from $v_i$, or $m \cdot s_{ij}$. (Since MF is invariant to scaling, we normalize $\matD$ w.l.o.g.)   
\end{proof}

Note that in struc2vec, increasing $m$ to sample more context reduces variance in $\matD$, but increasing $\ell$ simply causes the random walks to move further from the original node $v$ and sample context based on similarity to more structurally distant nodes.  Lemma \ref{sgns-mf} connects \embedding to a version of struc2vec with maximal $m$ and minimal $\ell$, further justifying its success by comparison.   

\end{document}